\documentclass[english,aps,prd,reprint]{revtex4}
\usepackage[T1]{fontenc}
\usepackage[latin9]{inputenc}
\usepackage[a4paper]{geometry}
\usepackage{babel}
\usepackage{wrapfig}
\usepackage{amsmath}
\usepackage{amsthm}
\usepackage{amssymb}
\usepackage{graphicx}
\usepackage[unicode=true,pdfusetitle,
 bookmarks=true,bookmarksnumbered=false,bookmarksopen=false,
 breaklinks=false,pdfborder={0 0 1},backref=false,colorlinks=false]
 {hyperref}
\usepackage{breakurl}

\makeatletter
\@ifundefined{textcolor}{}
{%
 \definecolor{BLACK}{gray}{0}
 \definecolor{WHITE}{gray}{1}
 \definecolor{RED}{rgb}{1,0,0}
 \definecolor{GREEN}{rgb}{0,1,0}
 \definecolor{BLUE}{rgb}{0,0,1}
 \definecolor{CYAN}{cmyk}{1,0,0,0}
 \definecolor{MAGENTA}{cmyk}{0,1,0,0}
 \definecolor{YELLOW}{cmyk}{0,0,1,0}
}
  \theoremstyle{definition}
  \newtheorem{defn}{\protect\definitionname}
  \theoremstyle{plain}
  \newtheorem{prop}{\protect\propositionname}
  \theoremstyle{plain}
  \newtheorem*{thm*}{\protect\theoremname}

\makeatother

  \providecommand{\definitionname}{Definition}
  \providecommand{\propositionname}{Proposition}
  \providecommand{\theoremname}{Theorem}

\begin{document}

\title{Singularities and Conjugate Points in FLRW Spacetimes}

\author{Huibert het Lam}
\email[E-mail: ]{h.hetlam@uu.nl}
\author{Tomislav Prokopec}
\email[E-mail: ]{t.prokopec@uu.nl}

\address{Institute for Theoretical Physics, Spinoza Institute and $EMME\Phi$,
Utrecht University,}

\address{Postbus 80.195, 3508 TD Utrecht, The Netherlands}
\begin{abstract}
Conjugate points play an important role in the proofs of the singularity
theorems of Hawking and Penrose. We examine the relation between singularities
and conjugate points in FLRW spacetimes with a singularity. In particular
we prove a theorem that when a non-comoving, non-spacelike geodesic
in a singular FLRW spacetime obeys conditions (\ref{eq:condition conjugate point at singularity})
and (\ref{eq:condition conjugate point at singularity 2}), every
point on that geodesic is part of a pair of conjugate points. The
proof is based on the Raychaudhuri equation. We find that the theorem
is applicable to all non-comoving, non-spacelike geodesics in FLRW
spacetimes with non-negative spatial curvature and scale factors that
near the singularity have power law behavior or power law behavior
times a logarithm. When the spatial curvature is negative, the theorem
is applicable to a subset of these spacetimes.
\end{abstract}
\maketitle

\section{Introduction}

Hawking and Penrose proved that under very general physical conditions
a spacetime has a singularity \citep{Hawking:1969sw,hawking1973large}.
A singularity is defined as a non-spacelike geodesic that is incomplete.
One uses this definition because test particles move on these trajectories
and thus have only traveled for a finite proper time. The rough idea
of the proof of these singularity theorems is that one assumes that
all non-spacelike geodesics are complete, one has a trapped surface
(e.g. an event horizon of a black hole) and that the weak energy condition
is obeyed. Under these conditions it is shown that there must be conjugate
points (which one can see as geodesics that are converging at two
points) and that leads to a contradiction with the completeness of
geodesics. Because of this relation between a singularity and conjugate
points, it is sometimes said that having a singularity is equivalent
to having conjugate points. This of course does not follow in any
way from the theorems of Hawking and Penrose, but one can examine
this idea in simple spacetimes. More specifically, we will do this
in FLRW spacetimes, which describe isotropic and spatially homogeneous
universes and are used as a model in cosmology. We prove that in a
singular FLRW spacetime, every point on a non-comoving, non-spacelike
geodesic satisfying conditions (\ref{eq:condition conjugate point at singularity})
and (\ref{eq:condition conjugate point at singularity 2}) is conjugate
to another point of that geodesic. We will examine the conditions
of the theorem for power law and logarithmic behavior of the scale
factor. This includes in particular an FLRW spacetime with flat spatial
three-surfaces that contains either a perfect homogeneous radiation
fluid or a perfect homogeneous matter fluid. We show that the theorem
is applicable to spacetimes with these scale factors when they have
non-negative spatial curvature. For negative spatial curvature, the
theorem is applicable to a subset of these scale factors. 

In this paper we first give a brief recap of the theory used to study
conjugate points, in particular the Raychaudhuri equation. We then
use this equation to prove the theorem after which we show that the
conditions of the theorem are obeyed for certain physical FLRW spacetimes
with non-negative spatial curvature and for a subset of the physical
spacetimes with negative spatial curvature. We adopt units in which
the velocity of light $c=1.$

\section{Conjugate Points and the Raychaudhuri Equation}

In this section we will recall the theory that is needed to study
conjugate points. We will state two propositions without proofs, one
can find these in e.g. \citep{hawking1973large,beem1996global}. Everything
will be stated for a general spacetime $(M,g),$ where $g$ is a Lorentzian
metric. The Riemann curvature tensor $R$ is defined by

\begin{equation}
R(X,Y)Z=\nabla_{X}(\nabla_{Y}Z)-\nabla_{Y}(\nabla_{X}Z)-\nabla_{[X,Y]}Z,
\end{equation}
where $\nabla$ is the Levi-Civita connection. Contracting the curvature
tensor, one can define the Ricci tensor $\mathrm{Ric}$ as the tensor
with components
\begin{equation}
R_{\lambda\nu}=R_{\;\;\lambda\rho\nu}^{\rho}.
\end{equation}
Covariant differentiation along a curve $\gamma(\tau)$ will be denoted
by $D_{\tau}$.
\begin{defn}
Let $\gamma:\:[\tau_{\mathrm{i}},\tau_{\mathrm{f}}]\rightarrow M$
be a non-spacelike geodesic segment. A \textit{variation through geodesics}
of $\gamma$ is a smooth function $\Gamma:\;(-\epsilon,\epsilon)\times[\tau_{\mathrm{i}},\tau_{\mathrm{f}}]\rightarrow M$,
such that $\Gamma(0,\tau)=\gamma(\tau)$ for all $\tau\in[\tau_{\mathrm{i}},\tau_{\mathrm{f}}]$
and every curve $\Gamma(w_{0},\tau)$ is a geodesic segment. The \textit{variation
field }of $\Gamma$ is the vector field $J(\tau)=\partial_{w}\Gamma(w,\tau)|_{w=0}$
along $\gamma$. 
\end{defn}
$ $
\begin{defn}
A vector field $J$ along a geodesic segment $\gamma:\;[\tau_{\mathrm{i}},\tau_{\mathrm{f}}]\rightarrow M$
that satisfies the Jacobi equation,
\begin{equation}
D_{\tau}^{2}J+R(J,\dot{\gamma})\dot{\gamma}=0,\label{eq:Jacobi equation}
\end{equation}
where $\dot{\gamma}=D_{\tau}\gamma,$ will be called a \textit{Jacobi
field}.
\end{defn}
\begin{prop}
Variations through geodesics $\Gamma$ of a geodesic segment $\gamma$
have a Jacobi field as variation field and every Jacobi field along
$\gamma$ corresponds to a variation through geodesics. 
\end{prop}
From Eq. (\ref{eq:Jacobi equation}) it then follows that the Jacobi
fields along a geodesic segment are determined by initial conditions
$J(\tau_{0})$ and $D_{\tau}J(\tau_{0})$ at $\gamma(\tau_{0})$ and
thus form an eight dimensional subspace of the space of vector fields
along $\gamma.$ 

\subsection{Timelike Geodesic Segment}

We now first restrict to timelike geodesic segments $\gamma$. 
\begin{defn}
If $\gamma$ is a timelike geodesic segment joining $p,q\in\gamma$,
$p$ is said to be \textit{conjugate} to $q$ along $\gamma$ if there
exists a non-vanishing Jacobi field $J$ along $\gamma$ such that
$J$ is zero at $p$ and $q$.
\end{defn}
\begin{wrapfigure}{r}{0.4\columnwidth}%
\begin{centering}
\includegraphics[scale=0.65]{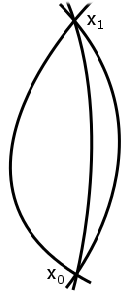}
\par\end{centering}
\caption{\label{fig:conjugate points}Let the lines be geodesics of some two-dimensional
manifold. Then $x_{0}$ and $x_{1}$ are conjugate to each other.}

\end{wrapfigure}%
 In figure \ref{fig:conjugate points} one can find an illustration
to give some idea of conjugate points. Jacobi fields corresponding
to conjugate points necessarily live in $N(\gamma),$ the space of
vector fields along $\gamma$ that are orthogonal to $\dot{\gamma}.$
In the theorem that we prove in section \ref{subsec:The-Theorem},
our goal is to show that there exists a conjugate point to a certain
fixed point on a geodesic. To do this we only have to look at the
Jacobi fields that are orthogonal to $\dot{\gamma}$ and vanish at
this initial point. This means that we can restrict to a three dimensional
vector space of Jacobi fields. To describe all of these fields at
once, we will introduce Jacobi tensors. We first simplify notation
by defining 
\begin{eqnarray}
R_{\gamma}(v) & = & R\left(v,\dot{\gamma}(\tau)\right)\dot{\gamma}(\tau).
\end{eqnarray}
Let $A:\;N(\gamma)\rightarrow N(\gamma)$ be a smooth tensor field.
Since $g\left(R_{\gamma}\left(v\right),\dot{\gamma}(\tau)\right)=0$
we can define a map $R_{\gamma}A:\:N(\gamma(\tau))\rightarrow N(\gamma(\tau))$
by
\begin{eqnarray}
R_{\gamma}A(\tau)(v) & = & R_{\gamma}\left(A(\tau)(v)\right).
\end{eqnarray}

\begin{defn}
A smooth $(1,1)$ tensor field $A:\;N(\gamma)\rightarrow N(\gamma)$
is called a \textit{Jacobi tensor} \textit{field} if it satisfies
\begin{eqnarray}
D_{\tau}^{2}A+R_{\gamma}A & = & 0,\label{eq:Jacobi Tensor field 1}\\
\mathrm{Ker}(A(\tau))\cap\mathrm{Ker}(D_{\tau}A(\tau)) & = & \{0\}\label{eq:Jacobi tensor field 2}
\end{eqnarray}
for all $\tau\in[\tau_{\mathrm{i}},\tau_{\mathrm{f}}]$. Here $\mathrm{Ker}(A(\tau))$
is the kernel of $A(\tau)$.
\end{defn}
If $V\in N(\gamma)\backslash\{0\}$ is a parallel transported vector
field along $\gamma$, i.e. $D_{\tau}V=0,$ and $A(\tau)$ a Jacobi
tensor field, define $J(\tau)=A(\tau)V(\tau)$. Then $J(\tau)$ is
a Jacobi field. Condition (\ref{eq:Jacobi tensor field 2}) guarantees
that $J$ is non-trivial. Therefore $A$ can be seen as describing
different families of geodesics at the same time. We now define a
Jacobi tensor field that describes all solutions to Eq. (\ref{eq:Jacobi equation})
living in $N(\gamma)$ and that vanish at $\gamma(\tau_{\mathrm{i}})$: 
\begin{defn}
Let $\{E_{\mu}\}$, $\mu=0,1,2,3$ be a parallel transported orthonormal
frame along $\gamma$ such that $E_{0}=\dot{\gamma}$. Let $J_{i}(\tau),$
$i\in\{1,2,3\},$ be the Jacobi field with $J_{i}(\tau_{\mathrm{i}})=0$
and $D_{\tau}J_{i}(\tau_{\mathrm{i}})=E_{i}(\tau_{\mathrm{i}})$.
Let $A$ be the tensor such that the components in the basis $E_{\mu}$
are given by 
\begin{eqnarray}
A_{\;\;l}^{k}(\tau) & = & (J_{l}(\tau))^{k};\label{eq:jacobi tensor related to point}\\
A_{\;\;0}^{0}=A_{\;\;0}^{k}=A_{\;\;l}^{0} & = & 0,\nonumber 
\end{eqnarray}
for $k,l=1,2,3$. 
\end{defn}
If $A$ is singular for some $\tau$ this will correspond to a Jacobi
field that vanishes at that $\gamma(\tau),$ hence that point on the
geodesic is conjugate to $\gamma(\tau_{\mathrm{i}}).$ So points conjugate
to $\gamma(\tau_{\mathrm{i}})$ are the points where $\det A=0.$
To examine whether $A$ is singular at some point, we develop some
more machinery.
\begin{defn}
Let $B_{A}=\left(D_{\tau}A\right)A^{-1}$ at points where $\det A\neq0$
\end{defn}
\begin{enumerate}
\item The \textit{expansion} $\theta_{A}$ is 
\begin{equation}
\theta_{A}=\mathrm{tr}(B_{A}).
\end{equation}
\item The \textit{vorticity tensor} $\omega_{A}$ is 
\begin{equation}
\omega_{A}=\frac{1}{2}(B_{A}-B_{A}^{\dagger}).
\end{equation}
\item The \textit{shear tensor} $\sigma_{A}$ is 
\begin{equation}
\sigma_{A}=\frac{1}{2}(B_{A}+B_{A}^{\dagger})-\frac{\theta_{A}}{3}I,
\end{equation}
where $I$ is the identity matrix. 
\end{enumerate}
Notice that 
\begin{equation}
B_{A}=\omega_{A}+\sigma_{A}+\frac{\theta_{A}}{3}I.\label{eq:expansion B}
\end{equation}

\begin{prop}
\textup{The vorticity 
\begin{equation}
\omega_{A}=0,
\end{equation}
the expansion
\begin{equation}
\theta_{A}=\partial_{\tau}\log\left(\det A\right)\label{eq:expansion in A}
\end{equation}
and the derivative of $\theta_{A}$ is given by
\begin{equation}
\dot{\theta}_{A}=-\mathrm{Ric}(\dot{\gamma}(\tau),\dot{\gamma}(\tau))-\mathrm{tr}(\sigma_{A}^{2})-\frac{\theta_{A}^{2}}{3}.\label{eq:generalized Raychaudhuri eq timelike}
\end{equation}
Eq. (\ref{eq:generalized Raychaudhuri eq timelike}) is also called
the Raychaudhuri equation for timelike geodesics.}
\end{prop}

\subsection{Null Geodesic Segment}

For null geodesic segments $\gamma:\;[\tau_{\mathrm{i}},\tau_{\mathrm{f}}]\rightarrow M$
one can also define conjugate points using Jacobi fields. However
now $\dot{\gamma}\in N(\gamma)$, so since we will be interested in
the convergence of geodesics, it makes more sense to look at the projection
of Jacobi fields to a quotient space formed by identifying vectors
that differ by a multiple of $\dot{\gamma}.$ This idea is implicitly
used in \citep{hawking1973large} and further developed in \citep{Bolts}.
One can do a similar analysis as for the timelike case and define
Jacobi classes $\bar{J}$, Jacobi tensor fields $\bar{A}$, vorticity
$\bar{\omega}_{\bar{A}},$ shear $\bar{\sigma}_{\bar{A}}$ and expansion
$\bar{\theta}_{\bar{A}}$ for this quotient space. One can derive
that 
\begin{equation}
\bar{\theta}_{\bar{A}}=\partial_{\tau}\log\left(\det\bar{A}\right)\label{eq:expansion in A null}
\end{equation}
 and derive a Raychaudhuri equation which is given by 
\begin{equation}
\partial_{\tau}\bar{\theta}_{\bar{A}}=-\mathrm{Ric}(\dot{\gamma},\dot{\gamma})-\mathrm{tr}(\bar{\sigma}_{\bar{A}}^{2})-\frac{\bar{\theta}_{\bar{A}}^{2}}{2}.\label{eq:raychaudhuri eq null geodesic}
\end{equation}
Points conjugate to $\gamma(\tau_{\mathrm{i}})$ correspond to points
where $\det\bar{A}=0,$ with $\bar{A}$ a specific Jacobi tensor field
constructed for a null geodesic as $A$ in (\ref{eq:jacobi tensor related to point})
for a timelike geodesic.

\section{FLRW Spacetimes\label{sec:Relation-between-Conjugate}}

We would like to study the relation between conjugate points and a
singularity in a spacetime with an FLRW metric. This metric describes
a spatially homogeneous, isotropic spacetime and in spherical coordinates
it is given by: 
\begin{equation}
ds^{2}=-dt^{2}+a(t)^{2}\left[\frac{dr^{2}}{1-\kappa r^{2}}+r^{2}\left(d\theta^{2}+\sin^{2}(\theta)d\varphi^{2}\right)\right],\label{eq:flrw metric}
\end{equation}
where $\kappa$ is the curvature of spacelike three-surfaces and the
scale factor $a(t)$ is normalized such that $a(t_{1})=1$ for some
time $t_{1}$. This metric is a good description of our universe,
since from experiments as WMAP and Planck, it follows that our universe
is spatially homogeneous and isotropic when averaged over large scales.
Geodesics $\gamma(\tau)$, where $\tau$ is an affine parameter, satisfy
\begin{equation}
\frac{d\gamma^{0}}{d\tau}=\frac{\sqrt{C-\epsilon_{\mathrm{n}}a^{2}}}{a},\label{eq:time component geodesic}
\end{equation}
where $C=|\vec{V}(t_{1})|^{2}=g_{ij}\dot{\gamma}^{i}\dot{\gamma}^{j}(t_{1})$
and $\epsilon_{\mathrm{n}}$ is the normalization of the geodesic:
$\epsilon_{\mathrm{n}}=0$ for null geodesics and $\epsilon_{\mathrm{n}}=-1$
for timelike geodesics. As argued in \citep{Lam:2016kmt}, singularities
(which are in general defined as incomplete non-spacelike geodesics)
in this spacetime are the points where the scale factor $a$ vanishes.
We would like to prove that under certain conditions a singularity
implies that all points on a geodesic are part of a pair of conjugate
points. 

\subsection{Examining the Definition of Conjugate Points}

Since the metric (\ref{eq:flrw metric}) becomes degenerate at a singularity
we can try to generalize the definition of conjugate points to also
include these points. For the theorem this will be important because
we will only be able to show that a certain point $\gamma(t)$ is
conjugate to a point $\gamma(t')$ where $t_{0}\leq t'<t$ and $t_{0}$
corresponds to the singularity $a(t_{0})=0$. Let us examine a specific
model. We will study the FLRW metric with $\kappa=0$ and $a(t)=\sqrt{t}.$
This models a spatially flat universe with a perfect radiation fluid
for which the energy density $\rho\propto1/a^{4}$ and is consistent
with current observations \citep{Ade:2015xua}. To derive the geodesics
we use Cartesian coordinates for the metric (\ref{eq:flrw metric})
\begin{equation}
ds^{2}=-dt^{2}+a(t)^{2}\left(dx^{i}\right)^{2}.
\end{equation}
Let a geodesic be given by $\gamma(\tau)=\left(t(\tau),x^{i}(\tau)\right)$
and let $u^{\mu}=d\gamma^{\mu}/d\tau$. The geodesic equations are
given by
\begin{eqnarray}
\frac{du^{0}}{d\tau}+a\dot{a}\left(u^{i}\right)^{2} & = & 0\nonumber \\
\frac{du^{i}}{d\tau}+2\frac{\dot{a}}{a}u^{0}u^{i} & = & 0.
\end{eqnarray}
The second equation can be rewritten as
\begin{equation}
\frac{d}{d\tau}\left[a^{2}u^{i}\right]=0
\end{equation}
with solution 
\begin{equation}
u^{i}=\frac{C_{i}}{a^{2}},\label{eq:derivative normal coordinates to tau}
\end{equation}
where $C_{i}$ are constants. The constraint equation is
\begin{equation}
\epsilon_{\mathrm{n}}=-\left(u^{0}\right)^{2}+a^{2}\left(u^{i}\right)^{2}
\end{equation}
and leads to 
\begin{equation}
u^{0}=\frac{\sqrt{C-\epsilon_{\mathrm{n}}a^{2}}}{a},\label{eq:time derived to tau}
\end{equation}
(Eq. (\ref{eq:time component geodesic})) where $C=\sum_{i}C_{i}^{2}$.
Let us now consider timelike geodesics, $\epsilon_{\mathrm{n}}=-1$
and choose $C=1$. We can then solve Eq. (\ref{eq:time derived to tau})
for $a(t)=\sqrt{t}$ by
\begin{equation}
\sqrt{t+t^{2}}-\sinh^{-1}\left(\sqrt{t}\right)=\tau,\label{eq:time in terms of tau}
\end{equation}
where we chose $\tau$ such that $\tau=0$ at the singularity. From
Eqs. (\ref{eq:derivative normal coordinates to tau}) and (\ref{eq:time derived to tau})
we find that 
\begin{equation}
\frac{dx^{i}}{dt}=\frac{1}{\sqrt{1+a^{2}}}\frac{C_{i}}{a},
\end{equation}
which is solved by
\begin{eqnarray}
x^{i} & = & 2C_{i}\sinh^{-1}\left(\sqrt{t}\right)+D_{i},\label{eq:solution geodesic eq}
\end{eqnarray}
where $D_{i}$ are constants (notice that we have the restriction
$1=\sum_{i}C_{i}^{2}$). We consider the geodesic
\begin{equation}
\gamma=\left(t,2\sinh^{-1}\left(\sqrt{t}\right),0,0\right)
\end{equation}
and we want to examine conjugate points along this geodesic. We now
construct the matrix $A$ of Eq. (\ref{eq:jacobi tensor related to point})
corresponding to the point $\gamma(t_{2})$ for this geodesic. An
orthonormal basis that is parallel transported along this geodesic
is given by 
\begin{eqnarray}
E_{0} & = & \left(\sqrt{\frac{1+t}{t}},\frac{1}{t},0,0\right)\nonumber \\
E_{1} & = & \left(\frac{1}{\sqrt{t}},\frac{\sqrt{1+t}}{t},0,0\right)\nonumber \\
E_{2} & = & \left(0,0,\frac{1}{\sqrt{t}},0\right)\nonumber \\
E_{3} & = & \left(0,0,0,\frac{1}{\sqrt{t}}\right).
\end{eqnarray}
We now need the Jacobi fields $J_{i}$ for $i\in\{1,2,3\}$ such that
$J_{i}(t_{2})=0$ and $D_{\tau}J_{i}(t_{2})=E_{i}(t_{2})$. The differential
equations (\ref{eq:Jacobi equation}) for the first 2 components of
the Jacobi fields only depend on each other. The differential equation
for $J_{i}^{k}$, $k\in\{2,3\}$ is given by
\begin{equation}
\frac{(1+2t)\left(J_{i}^{k}\right)'+2t(1+t)\left(J_{i}^{k}\right)''}{t}=0.\label{eq:jacobi field diff eq}
\end{equation}
This implies that 
\begin{eqnarray}
J_{1} & = & \left(h_{1}(t),h_{2}(t),0,0\right)\nonumber \\
J_{2} & = & \left(0,0,h_{3}(t),0\right)\label{eq:jacobi fields cartesian coordinates}\\
J_{3} & = & \left(0,0,0,h_{3}(t)\right).\nonumber 
\end{eqnarray}
We can solve Eq. (\ref{eq:jacobi field diff eq}) for $h_{3}$ explicitly
and find 
\begin{equation}
h_{3}(t)=-2\sqrt{t_{2}}\left(\sinh^{-1}\left(\sqrt{t_{2}}\right)-\sinh^{-1}\left(\sqrt{t}\right)\right).\label{eq:jacobi field 3}
\end{equation}
We have to solve for $h_{1}$ and $h_{2}$ numerically. The matrix
$A$ is then given by
\begin{equation}
A=\left(\begin{array}{ccc}
-\frac{1}{\sqrt{t}}h_{1}+\sqrt{1+t}h_{2} & 0 & 0\\
0 & \sqrt{t}h_{3} & 0\\
0 & 0 & \sqrt{t}h_{3}
\end{array}\right),
\end{equation}
which has determinant
\begin{equation}
\det A=t\left(-\frac{1}{\sqrt{t}}h_{1}+\sqrt{1+t}h_{2}\right)h_{3}^{2}.
\end{equation}
Notice that at $t=0$, $\sqrt{t}h_{3}(t)=0$, which naively would
mean that $\gamma(t_{2})$ is conjugate to the point at the singularity.
However, in other coordinate systems the Jacobi field does not vanish
(see Eq. (\ref{eq:jacobi field 3})). This behavior is caused by the
degeneracy of the metric at the singularity. The norm of the Jacobi
field however, is zero in both coordinate systems.

We use this example as a motivation to generalize the definition of
a conjugate point to include points where the metric is degenerate.
From a physical point of view it is the norm of the Jacobi field that
matters since this corresponds to the distance between particles moving
on nearby geodesics. That is why we will also say that we have conjugate
points on a timelike geodesic when the norm of the Jacobi field vanishes.
Such a Jacobi field should still be perpendicular to the geodesic
(otherwise one could just get that it is a null vector). As long as
the metric is non-degenerate this definition is the same as our original
definition. Notice that the vanishing of the determinant of the matrix
$A$ is equivalent to a Jacobi field $J$ perpendicular to $\dot{\gamma}$
and such that $g(E_{i},J)=0$ for all $i$. From 
\begin{equation}
g(J,J)=\sum_{i}g(E_{i},J)^{2}
\end{equation}
we conclude that $g(E_{i},J)=0$ for all $i$ is equivalent to $g(J,J)=0.$ 

With this new definition two points on a geodesic can be conjugate
in two different ways. The first one is that geodesics are indeed
converging to one point (to first order), the second one is that that
does not happen, but that the norm of the Jacobi field vanishes. We
found that the vanishing of the determinant of $A$ is equivalent
to this new definition if the Jacobi field is perpendicular to the
geodesic. In the same way one can give a generalized definition of
conjugate points for null geodesics.

\subsection{The Theorem\label{subsec:The-Theorem}}

We will now prove the theorem that states that when a certain non-comoving,
non-spacelike geodesic satisfies conditions (\ref{eq:condition conjugate point at singularity})
and (\ref{eq:condition conjugate point at singularity 2}), every
point on that geodesic is part of a pair of conjugate points. Here
we do not know whether geodesics actually converge to that point.
To prove this for a point $\gamma(t_{2})$ on a timelike geodesic,
the idea is to use Eqs. (\ref{eq:expansion in A}) and (\ref{eq:generalized Raychaudhuri eq timelike})
to derive an inequality for $\log\left(\det A(t)\right).$ From this
inequality we show that $\log\left(\det A(t)\right)$ goes to $-\infty$
at a point $\gamma(t')$ that lies in between the singularity and
$\gamma(t_{2}).$ This means that $\det A(t')=0$ which implies that
$\gamma(t')$ is conjugate to $\gamma(t_{2}).$ For null-geodesics
we use the same strategy using Eqs. (\ref{eq:expansion in A null})
and (\ref{eq:raychaudhuri eq null geodesic}).

. 
\begin{thm*}
\label{proof conjugate point using raychaudhuri}Let $\gamma(\tau(t))$
be a non-comoving ($C>0$), non-spacelike geodesic in a spacetime
with FLRW metric such that $a(t_{0})=0$ for a certain $t_{0}$ and
$a$ is smooth for $t>t_{0}$. Let 
\begin{equation}
f(t)=3\overset{..}{a}+2\frac{C}{a}\left[\frac{\overset{..}{a}}{a}-\frac{\dot{a}^{2}}{a^{2}}-\frac{\kappa}{a^{2}}\right]
\end{equation}
and define
\begin{eqnarray}
f_{+}(t) & = & \begin{cases}
f(t) & \;\mbox{for t where }f(t)\geq0\\
0 & \;\mbox{for t where }f(t)<0
\end{cases}\\
f_{-}(t) & = & \begin{cases}
-f(t) & \;\mbox{for t where }f(t)\leq0\\
0 & \;\mbox{for t where }f(t)>0.
\end{cases}\nonumber 
\end{eqnarray}
A point $\gamma(\tau(t_{2}))$ for $t_{2}\neq t_{0}$ is conjugate
to a point $\gamma(\tau(t'))$ where $t_{0}\leq t'<t_{2}$ if the
following conditions are satisfied: 
\begin{eqnarray}
\lim_{t\rightarrow t_{0}}\int_{t}^{t_{1}}a(t')\int_{t'}^{t_{1}}\frac{1}{a}\left[\frac{\overset{..}{a}}{a}-\frac{\dot{a}^{2}}{a^{2}}-\frac{\kappa}{a^{2}}\right]dt''dt' & = & -\infty\label{eq:condition conjugate point at singularity}\\
\lim_{t\rightarrow t_{0}}\int_{t}^{t_{1}}f_{+}dt' & = & \alpha\in\mathbb{R}_{\geq0}\label{eq:condition conjugate point at singularity 2}
\end{eqnarray}
for a $t_{1}>t_{0}.$ 
\end{thm*}
\begin{proof}
We will prove this separately for timelike and null geodesics. Let
$\gamma$ be a timelike geodesic. Let $\gamma(\tau(t_{2}))$ be a
point on this geodesic and let $A$ denote the Jacobi tensor field
as defined in (\ref{eq:jacobi tensor related to point}). To show
that $\gamma(\tau(t_{2}))$ is conjugate to a point $\gamma(\tau(t'))$
with $t_{0}\leq t'<t$, we will show that $\log\left(\det A\right)$
has to go to $-\infty$ at some point $\gamma(t').$ 

Consider Eq. (\ref{eq:time component geodesic}) and the Raychaudhuri
equation, Eq. (\ref{eq:generalized Raychaudhuri eq timelike}). Since
$\sigma_{A}$ is symmetric we have that $\mathrm{tr}\left(\sigma_{A}^{2}\right)$
is positive such that
\begin{equation}
\frac{d\theta_{A}}{dt}=\frac{d\tau}{dt}\frac{d\theta_{A}}{d\tau}\leq\frac{-a}{\sqrt{C+a^{2}}}\mathrm{Ric}(\dot{\gamma}(\tau),\dot{\gamma}(\tau)).
\end{equation}
For the FLRW metric we then find that
\begin{equation}
-\mathrm{Ric}(\dot{\gamma}(\tau),\dot{\gamma}(\tau))=3\frac{\overset{..}{a}}{a}+2\frac{C}{a^{2}}\left[\frac{\overset{..}{a}}{a}-\frac{\dot{a}^{2}}{a^{2}}-\frac{\kappa}{a^{2}}\right],
\end{equation}
which results in 
\begin{eqnarray}
\frac{d\theta_{A}}{dt} & \leq & \frac{1}{\sqrt{C+a^{2}}}\left(f_{+}-f_{-}\right).\label{eq:raychaudhuri eq flrw metric}
\end{eqnarray}
Notice that in conditions (\ref{eq:condition conjugate point at singularity})
and (\ref{eq:condition conjugate point at singularity 2}) we can
assume that $t_{1}<t_{2}$ and that $\gamma(t_{1})$ is not conjugate
to $\gamma(t_{2}).$ We find from condition (\ref{eq:condition conjugate point at singularity})
that: 
\begin{equation}
\lim_{t\rightarrow t_{0}}\int_{t}^{t_{1}}a\int_{t'}^{t_{1}}f_{+}dt''dt'-\lim_{t\rightarrow t_{0}}\int_{t}^{t_{1}}a\int_{t'}^{t_{1}}f_{-}dt''dt'=-\infty.\label{eq:positive and negative parts}
\end{equation}
From condition (\ref{eq:condition conjugate point at singularity 2})
it follows that
\begin{equation}
\int_{t}^{t_{1}}f_{+}dt'
\end{equation}
is a function that is $\alpha$ at $t_{0},$ $0$ at $t_{1}$ and
strictly decreasing. Hence
\begin{equation}
a(t)\int_{t}^{t_{1}}f_{+}dt'
\end{equation}
 is vanishing at $t_{0}$ and $t_{1}$ and continuous and positive
in between. This implies that
\begin{equation}
\lim_{t\rightarrow t_{0}}\int_{t}^{t_{1}}a\int_{t'}^{t_{1}}f_{+}dt''dt'=\beta\in\mathbb{R}_{\geq0}\label{eq:limits pos en neg 2}
\end{equation}
 and together with Eq. (\ref{eq:positive and negative parts}) this
gives 
\begin{eqnarray}
\lim_{t\rightarrow t_{0}}\int_{t}^{t_{1}}a\int_{t'}^{t_{1}}f_{-}dt''dt' & = & \infty.\label{eq:limits pos and neg}
\end{eqnarray}
For $t<t_{1}$ we find with Eq. (\ref{eq:raychaudhuri eq flrw metric})
that 
\begin{eqnarray}
\theta_{A}(t) & = & -\int_{t}^{t_{\text{1}}}\frac{d\theta_{A}}{dt}dt'+\theta_{A}(t_{1})\nonumber \\
 & \geq & -\frac{1}{\sqrt{C}}\int_{t}^{t_{1}}f_{+}dt'+\frac{1}{\sqrt{C+a_{\mathrm{max}}^{2}}}\int_{t}^{t_{1}}f_{-}dt'+\theta_{A}(t_{1}),
\end{eqnarray}
where $a_{\mathrm{max}}=\max\{a(t)|t_{0}\leq t\leq t_{1}\}.$ Then
using Eq. (\ref{eq:expansion in A}): 
\begin{eqnarray}
\log\left(\det A(t)\right) & = & -\int_{t}^{t_{1}}\frac{a}{\sqrt{C+a^{2}}}\theta_{A}dt'+\log\left(\det A(t_{1})\right)\nonumber \\
 & \leq & \frac{1}{C}\int_{t}^{t_{1}}a\int_{t'}^{t_{1}}f_{+}dt''dt'-\frac{1}{C+a_{\mathrm{max}}^{2}}\int_{t}^{t_{1}}a\int_{t'}^{t_{1}}f_{-}dt''dt'\nonumber \\
 &  & -\theta_{A}(t_{1})\int_{t}^{t_{1}}\frac{a}{\sqrt{C+a^{2}}}dt'+\log\left(\det A(t_{1})\right).\label{eq:expansion parameter intermediate}
\end{eqnarray}
With Eqs. (\ref{eq:limits pos en neg 2}) and (\ref{eq:limits pos and neg})
it then follows that the right-hand side of Eq. (\ref{eq:expansion parameter intermediate})
goes to $-\infty$ in the limit $t\rightarrow t_{0}$. That means
that $\gamma(\tau(t_{2}))$ is conjugate to a point $\gamma(\tau(t'))$
with $t_{0}\leq t'<t_{1}$. 

Consider now a null geodesic $\gamma$ and let $\gamma(\tau(t_{2}))$
be a point on this geodesic. The Raychaudhuri equation, Eq. (\ref{eq:raychaudhuri eq null geodesic}),
reads:
\begin{eqnarray}
\frac{d\bar{\theta}_{\bar{A}}}{dt} & = & \frac{a}{\sqrt{C}}\left(2\frac{C}{a^{2}}\left[\frac{\overset{..}{a}}{a}-\frac{\dot{a}^{2}}{a^{2}}-\frac{\kappa}{a^{2}}\right]-\mathrm{tr}\left(\bar{\sigma}_{\bar{A}}^{2}\right)-\frac{\bar{\theta}_{\bar{A}}^{2}}{2}\right)\nonumber \\
 & \leq & 2\frac{\sqrt{C}}{a}\left[\frac{\overset{..}{a}}{a}-\frac{\dot{a}^{2}}{a^{2}}-\frac{\kappa}{a^{2}}\right].
\end{eqnarray}
We can again assume that $t_{1}<t_{2}$ and that $\gamma(t_{1})$
is not conjugate to $\gamma(t_{2})$. It then follows that for $t<t_{1}$
\begin{eqnarray}
\bar{\theta}_{\bar{A}}(t) & = & -\int_{t}^{t_{1}}\frac{d\bar{\theta}_{\bar{A}}}{dt}dt'+\bar{\theta}_{\bar{A}}(t_{1})\nonumber \\
 & \geq & -2\sqrt{C}\int_{t}^{t_{1}}\frac{1}{a}\left[\frac{\overset{..}{a}}{a}-\frac{\dot{a}^{2}}{a^{2}}-\frac{\kappa}{a^{2}}\right]dt'+\bar{\theta}_{\bar{A}}(t_{1}).
\end{eqnarray}
This implies that
\begin{eqnarray}
\log\left(\det\bar{A}(t)\right) & = & -\int_{t}^{t_{1}}\frac{a}{\sqrt{C}}\bar{\theta}_{\bar{A}}dt'+\log\left(\det\bar{A}(t_{1})\right)\nonumber \\
 & \leq & 2\int_{t}^{t_{1}}a\int_{t'}^{t_{1}}\frac{1}{a}\left[\frac{\overset{..}{a}}{a}-\frac{\dot{a}^{2}}{a^{2}}-\frac{\kappa}{a^{2}}\right]dt''dt'-\bar{\theta}_{\bar{A}}(t_{1})\frac{1}{\sqrt{C}}\int_{t}^{t_{1}}adt'+\log\left(\det\bar{A}(t_{1})\right),\label{eq:log a bar}
\end{eqnarray}
where we have used Eq. (\ref{eq:expansion in A null}). Condition
(\ref{eq:condition conjugate point at singularity}) then implies
that the right-hand side of Eq. (\ref{eq:log a bar}) goes to $-\infty$
in the limit $t\rightarrow t_{0}.$ Hence $\gamma(\tau(t_{2}))$ is
conjugate to a point $\gamma(\tau(t'))$ where $t_{0}\leq t'<t_{1}$. 
\end{proof}
Notice that one can rewrite condition (\ref{eq:condition conjugate point at singularity})
by partially integrating the first term such that one obtains
\begin{eqnarray}
-\infty & = & \lim_{t\rightarrow t_{0}}\int_{t}^{t_{1}}a(t')\int_{t'}^{t_{1}}\frac{1}{a}\left[\frac{d}{dt}\frac{\dot{a}}{a}-\frac{\kappa}{a^{2}}\right]dt''dt'\nonumber \\
 & = & \frac{\dot{a}(t_{1})}{a^{2}(t_{1})}\int_{t_{0}}^{t_{1}}adt'-\log(a(t_{1}))+\lim_{t\rightarrow t_{0}}\log(a(t))+\lim_{t\rightarrow t_{0}}\int_{t}^{t_{1}}a(t')\int_{t'}^{t_{1}}\frac{\dot{a}^{2}-\kappa}{a^{3}}dt''dt'.
\end{eqnarray}
Thus condition (\ref{eq:condition conjugate point at singularity})
is definitely satisfied when 
\begin{equation}
\lim_{t\rightarrow t_{0}}\int_{t}^{t_{1}}a\int_{t'}^{t_{1}}\frac{\dot{a}^{2}-\kappa}{a^{3}}dt''dt'
\end{equation}
is not $\infty.$ 

Also condition (\ref{eq:condition conjugate point at singularity 2})
is satisfied as soon as $f$ is negative for $t\in(t_{0},t_{0}+\delta),$
$\delta\ll1$.

The theorem can be proven under different conditions. One set of such
conditions would be for instance
\begin{eqnarray}
\lim_{t\rightarrow t_{0}}\int_{t}^{t_{1}}\frac{a}{\sqrt{C+a^{2}}}\int_{t'}^{t_{1}}\frac{1}{\sqrt{C+a^{2}}}\left(3\overset{..}{a}+2\frac{C}{a}\left[\frac{\overset{..}{a}}{a}-\frac{\dot{a}^{2}}{a^{2}}-\frac{\kappa}{a^{2}}\right]\right)dt''dt' & = & -\infty;\\
\lim_{t\rightarrow t_{0}}\int_{t}^{t_{1}}a\int_{t'}^{t_{1}}\frac{1}{a}\left[\frac{\overset{..}{a}}{a}-\frac{\dot{a}^{2}}{a^{2}}-\frac{\kappa}{a^{2}}\right]dt''dt' & = & -\infty\nonumber 
\end{eqnarray}
which would have made the proof really easy.

\subsection{Relation to Physical Spacetimes}

The theorem is applicable to FLRW spacetimes with physically realistic
scale factors. Notice that the conditions of the theorem only depend
on the form of the scale factor $a(t)$ near the singularity. We will
assume it there to take one of the forms
\begin{eqnarray}
a(t) & = & t^{1/\epsilon}\label{eq:power law}\\
a(t) & = & -t^{1/\epsilon}\log(t),\label{eq:one loop corrected}
\end{eqnarray}
with $\epsilon>0$. Notice that in case of power law behavior (\ref{eq:power law}),
\begin{equation}
\epsilon=-\frac{\dot{H}}{H^{2}}
\end{equation}
is the principal slow roll parameter ($H=\dot{a}/a$ is the Hubble
parameter). When $0<\epsilon\ll1$ the scale factor (\ref{eq:power law})
corresponds to inflation, $\epsilon=3/2$ gives the scale factor of
an FLRW spacetime with $\kappa=0$ containing a perfect homogeneous
matter fluid and $\epsilon=2$ gives the scale factor of an FLRW spacetime
with $\kappa=0$ containing a perfect homogeneous radiation fluid.
The second form (\ref{eq:one loop corrected}) of the scale factor
is related to one loop corrections. When matter is integrated out
the effective action contains, up to boundary terms, terms $R^{2}\log(R/\mu^{2})$
and $W^{2}\log(R/\mu^{2}),$ where $W$ is the Weyl tensor and $\mu$
is an energy scale \citep{Zeldovich:1971mw,Gurovich:1979xg,Barvinsky:1987uw,Barvinsky:1990up}.
This motivates to examine scale factors that have more complicated
behavior near the singularity and that is why we also study the logarithmic
behavior (\ref{eq:one loop corrected}). Notice however that \citep{Zeldovich:1971mw,Gurovich:1979xg}
focussed mostly on anisotropic expansions which actually help to resolve
the singularity. We will however still examine the form (\ref{eq:one loop corrected})
since in the end it just serves as an example to what kind of scale
factors the theorem can be applied. 

We will consider both of the scale factors (\ref{eq:power law}) and
(\ref{eq:one loop corrected}) separately, starting with the power
law behavior. We find that for $\epsilon\notin\{1,3\}$

\begin{eqnarray}
\int_{t}^{t_{1}}a\int_{t'}^{t_{1}}\frac{1}{a}\left[\frac{\overset{..}{a}}{a}-\frac{\dot{a}^{2}}{a^{2}}-\frac{\kappa}{a^{2}}\right]dt''dt' & = & \frac{\alpha\epsilon}{1+\epsilon}\left(t_{1}^{1/\epsilon+1}-t{}^{1/\epsilon+1}\right)-\frac{1}{1+\epsilon}\log\left(\frac{t_{1}}{t}\right)+\left[\frac{\epsilon}{\epsilon-3}\frac{\epsilon}{2\epsilon-2}\frac{\kappa}{t^{2/\epsilon-2}}\right]_{t}^{t_{1}},\nonumber \\
\end{eqnarray}
where 
\begin{equation}
\alpha=\frac{1}{1+\epsilon}\frac{1}{t_{1}^{1/\epsilon+1}}-\frac{\epsilon}{\epsilon-3}\frac{\kappa}{t_{1}^{3/\epsilon-1}}.
\end{equation}
This scale factor obeys condition (\ref{eq:condition conjugate point at singularity})
when
\begin{itemize}
\item $0<\epsilon<1$ and $\kappa>0$;
\item $\epsilon>1$ or $\kappa=0$.
\end{itemize}
Similarly, one can show that condition (\ref{eq:condition conjugate point at singularity})
is satisfied for $\epsilon=3.$ For $\epsilon=1,$ condition (\ref{eq:condition conjugate point at singularity})
is only satisfied for $\kappa>-1.$

Also 
\begin{equation}
f(t)=3\frac{1}{\epsilon}(\frac{1}{\epsilon}-1)\frac{1}{t^{2-1/\epsilon}}-2C\left(\frac{1}{\epsilon}\frac{1}{t^{2+1/\epsilon}}+\frac{\kappa}{t^{3/\epsilon}}\right)
\end{equation}
which goes to $-\infty$ in the limit $t\rightarrow0$ for
\begin{itemize}
\item $0<\epsilon<1$ and $\kappa>0$;
\item $\epsilon>1$ or $\kappa=0$
\end{itemize}
and that implies that condition (\ref{eq:condition conjugate point at singularity 2})
is satisfied. When $\epsilon=1$ and $\kappa>-1$, $\lim_{t\rightarrow0}f(t)=-\infty$
such that condition (\ref{eq:condition conjugate point at singularity 2})
is satisfied.

Concluding, we can apply the theorem to all non-comoving, non-spacelike
geodesics in FLRW spacetimes with a scale factor with power law behavior
(\ref{eq:power law}) in the cases
\begin{itemize}
\item $0<\epsilon<1$ and $\kappa>0$;
\item $\epsilon>1$ or $\kappa=0$;
\item $\epsilon=1$ and $\kappa>-1.$
\end{itemize}
We will now focus on scale factors of the form (\ref{eq:one loop corrected}).
We find that
\begin{equation}
\frac{\overset{..}{a}}{a}-\frac{\dot{a}^{2}}{a^{2}}-\frac{\kappa}{a^{2}}=\frac{-\frac{1}{\epsilon}\left(\log t\right)^{2}-\log t-1}{t^{2}\left(\log t\right)^{2}}-\frac{\kappa}{t^{2/\epsilon}\left(\log t\right)^{2}}.\label{eq:expression1 log(T)}
\end{equation}
Since condition (\ref{eq:condition conjugate point at singularity})
only depends on the behavior in the limit $t\rightarrow0,$ we only
have to consider the dominating term of expression (\ref{eq:expression1 log(T)}).
In this limit 
\begin{equation}
\frac{\overset{..}{a}}{a}-\frac{\dot{a}^{2}}{a^{2}}-\frac{\kappa}{a^{2}}\rightarrow\begin{cases}
-\frac{\kappa}{t^{2/\epsilon}\left(\log t\right)^{2}} & 0<\epsilon<1\mbox{ and \ensuremath{\kappa\neq0}}\\
-\frac{1}{\epsilon t^{2}} & \epsilon\geq1\mbox{ or \ensuremath{\kappa=0.}}
\end{cases}
\end{equation}
Hence for $\epsilon\geq1$ or $\kappa=0$ we find that 
\begin{eqnarray}
\int_{t}^{t_{1}}a(t')\int_{t'}^{t_{1}}\frac{1}{a}\left[\frac{\overset{..}{a}}{a}-\frac{\dot{a}^{2}}{a^{2}}-\frac{\kappa}{a^{2}}\right]dt''dt' & \rightarrow & -\int_{t}^{t_{1}}t^{1/\epsilon}\log(t)\int_{t'}^{t_{1}}\left[\frac{1}{\epsilon}\frac{1}{t^{2+1/\epsilon}\log(t)}\right]dt''dt'\\
 & \rightarrow & -\infty.\nonumber 
\end{eqnarray}
For $0<\epsilon<1$ and $\kappa\neq0$ we find
\begin{equation}
\int_{t}^{t_{1}}a(t')\int_{t'}^{t_{1}}\frac{1}{a}\left[\frac{\overset{..}{a}}{a}-\frac{\dot{a}^{2}}{a^{2}}-\frac{\kappa}{a^{2}}\right]dt''dt'\rightarrow-\int_{t}^{t_{1}}t^{1/\epsilon}\log(t)\int_{t'}^{t_{1}}\left[\frac{\kappa}{t^{3/\epsilon}\left(\log t\right)^{3}}\right]dt''dt'
\end{equation}
which for $t\rightarrow0$ goes to $\infty$ when $\kappa<0$ and
goes to $-\infty$ when $\kappa>0.$ We conclude that condition (\ref{eq:condition conjugate point at singularity})
is obeyed in the cases 
\begin{itemize}
\item $0<\epsilon<1$ and $\kappa>0$;
\item $\epsilon\geq1$ or $\kappa=0.$
\end{itemize}
We consider now condition (\ref{eq:condition conjugate point at singularity 2}).
We have that
\begin{equation}
f(t)=-3\left(\frac{1}{\epsilon}(\frac{1}{\epsilon}-1)\log(t)+(\frac{2}{\epsilon}-1)\right)t^{1/\epsilon-2}+2\frac{C}{t^{1/\epsilon}\log(t)}\left[\frac{\frac{1}{\epsilon}\left(\log t\right)^{2}+\log t+1}{t^{2}\left(\log t\right)^{2}}+\frac{\kappa}{t^{2/\epsilon}\left(\log t\right)^{2}}\right]
\end{equation}
which goes for $t\rightarrow0$ to
\begin{equation}
f(t)\rightarrow\begin{cases}
\frac{2}{\epsilon}C\frac{1}{t^{2+1/\epsilon}\log t} & \mbox{ }\epsilon\geq1\mbox{ or \ensuremath{\kappa=0}};\\
2C\frac{\kappa}{t^{3/\epsilon}\left(\log t\right)^{3}} & \mbox{ }\mbox{0<\ensuremath{\epsilon<1\mbox{ and \ensuremath{\kappa\neq0}}.}}
\end{cases}
\end{equation}
Hence we find that $f\rightarrow-\infty$ such that condition (\ref{eq:condition conjugate point at singularity 2})
is obeyed in the cases
\begin{itemize}
\item $0<\epsilon<1$ and $\kappa>0;$
\item $\epsilon\geq1$ or $\kappa=0$,
\end{itemize}
which implies that the theorem is applicable to all non-comoving,
non-spacelike geodesics in exactly these cases.

\section{Conclusion}

We studied the connection between the occurrence of conjugate points
on geodesics and the existence of singularities in spacetimes with
an FLRW metric. In particular we proved that in a singular FLRW spacetime,
every point on a non-comoving, non-spacelike geodesic is part of a
pair of conjugate points if the geodesic satisfies conditions (\ref{eq:condition conjugate point at singularity})
and (\ref{eq:condition conjugate point at singularity 2}). To do
that we generalized the definition of conjugate points to include
points of the metric where it is degenerate. In the proof of the theorem
we extensively used the Raychaudhuri equation. We also showed that
the theorem is applicable to all non-comoving, non-spacelike geodesics
in FLRW spacetimes with a scale factor of the form 
\begin{eqnarray}
a(t) & = & t^{1/\epsilon}\label{eq:power law 2}
\end{eqnarray}
in the cases
\begin{itemize}
\item $0<\epsilon<1$ and $\kappa>0$;
\item $\epsilon>1$ or $\kappa=0$;
\item $\epsilon=1$ and $\kappa>-1,$ 
\end{itemize}
and for a scale factor of the form 
\begin{equation}
a(t)=-t^{1/\epsilon}\log(t)\label{eq:one loop 2}
\end{equation}
 in the cases
\begin{itemize}
\item $0<\epsilon<1$ and $\kappa>0;$
\item $\epsilon\geq1$ or $\kappa=0$.
\end{itemize}
The parameter $\epsilon$ is the principal slow roll parameter for
the form (\ref{eq:power law 2}) and $\kappa$ is the curvature of
spatial three-surfaces. Since the conditions of the theorem only depend
on the behavior of a scale factor near the singularity, we find that
for FLRW spacetimes that belong to one of these cases near the singularity,
every point on a non-comoving, non-spacelike geodesic belongs to a
pair of conjugate points. This includes in particular an FLRW spacetime
with flat spatial three-surfaces that contains either a perfect homogeneous
radiation fluid or a perfect homogeneous matter fluid. 

It would be of interest to examine the connection between conjugate
points and singularities further in FLRW spacetimes. One can also
study this connection for other metrics such as singular anisotropic
spacetimes and spacetimes containing a black hole.

\section*{Acknowledgments}

H. L. likes to thank Gil Cavalcanti for useful discussions about this
topic. This work was supported in part by the D-ITP consortium, a
program of the Netherlands Organization for Scientific Research (NWO)
that is funded by the Dutch Ministry of Education, Culture and Science
(OCW), and by the NWO Graduate Programme.

\bibliographystyle{elsarticle-num}
\bibliography{paper2}

\end{document}